\documentclass[11pt]{article}
\usepackage{amsfonts}
\usepackage{amsmath}

\newtheorem{theorem}{Theorem}

\newtheorem{corollary}[theorem]{Corollary}

\newtheorem{notation}[theorem]{Notation}

\newtheorem{proposition}[theorem]{Proposition}

\newenvironment{proof}[1][Proof]{\noindent\textbf{#1.} }{\ \rule{0.5em}{0.5em}}
\input{tcilatex}
\begin{document}

\title{Symplectic Covariance Properties for Shubin and Born--Jordan
Pseudo-Differential Operators}
\author{Maurice A. de Gosson \\
University of Vienna\\
Faculty of Mathematics - NuHAG\\
Nordbergstr. 15, 1090 Vienna}
\maketitle

\begin{abstract}
Among all classes of pseudo-differential operators only the Weyl operators
enjoy the property of symplectic covariance with respect to conjugation by
elements of the metaplectic group. In this paper we show that there is,
however, a weaker form of symplectic covariance for Shubin's $\tau $%
-dependent operators, in which the intertwiners no longer are metaplectic,
but still are invertible non-unitary operators. We also study the case of
Born--Jordan operators, which are obtained by averaging the $\tau $%
-operators over the interval $[0,1]$ (such operators have recently been
studied by Boggiatto and his collaborators). We show that metaplectic
covariance still hold for these operators, with respect top a subgroup of
the metaplectic group.
\end{abstract}

\section{Introduction}

In the early years of quantum mechanics physicists were confronted with an
ordering problem: assume that some quantization process associated to the
real variables $x$ (position) and $p$ (momentum) two operators $\widehat{X}$
and $\widehat{P}$ satisfying the canonical commutation rule $\widehat{X_{j}}%
\widehat{P_{j}}-\widehat{P_{j}}\widehat{X_{j}}=i\hbar$. What should then be
the operator associated to the monomial $x^{m}p^{n}$? The first to give a
mathematically motivated answer was Weyl \cite{Weyl}; he was was developing
his ideas on a group theoretical approach to quantization which lead to the
prescription \ 
\begin{equation}
x_{j}^{m}p_{j}^{\ell}\overset{\text{\textrm{Weyl}}}{\longrightarrow}\frac {1%
}{2^{\ell}}\sum_{k=0}^{\ell}%
\begin{pmatrix}
\ell \\ 
k%
\end{pmatrix}
\widehat{P_{j}}^{\ell-k}\widehat{X_{j}}^{m}\widehat{P_{j}}^{k}  \label{w3}
\end{equation}
It turns out that the Weyl ordering is a particular case of the more general
\textquotedblleft$\tau$-ordering\textquotedblright: for any real number $%
\tau $ one defines%
\begin{equation}
x_{j}^{m}p_{j}^{\ell}\overset{\tau}{\longrightarrow}\sum_{k=0}^{\ell}%
\begin{pmatrix}
\ell \\ 
k%
\end{pmatrix}
(1-\tau)^{k}\tau^{\ell-k}\widehat{P_{j}}^{k}\widehat{X_{j}}^{\ell }\widehat{%
P_{j}}^{\ell-k}  \label{tau}
\end{equation}
which reduces to Weyl's prescription when $\tau=\frac{1}{2}$. We will from
now assume that $\widehat{X_{j}}f=x_{j}f$ and $\widehat{P_{j}}f=-2\pi
i\partial_{x_{j}}f$. The $\tau$-ordering (\ref{tau}) is itself a particular
case of the Shubin pseudo-differential calculus (Shubin \cite{sh87}): given
a symbol $a$ the $\tau$-pseudo-differential operator $A_{\tau}=\limfunc{Op}%
_{\tau}(a)$ is formally defined by 
\begin{equation*}
A_{\tau}f(x)=\iint e^{2\pi ip(x-y)}a(\tau x+(1-\tau)y,p)f(y)dpdy;
\end{equation*}
for $\tau=\frac{1}{2}$ we recover the Weyl correspondence. Using Schwartz's
kernel theorem it is not difficult to show that for every continuous linear
operator $A:\mathcal{S}(\mathbb{R}^{n})\longrightarrow\mathcal{S}^{\prime }(%
\mathbb{R}^{n})$ and for every $\tau\in\mathbb{R}$ there exists $a\in%
\mathcal{S}^{\prime}(\mathbb{R}^{2n})$ such that $A=\limfunc{Op}_{\tau}(a)$;
the $\tau$-operators are thus of a very general nature. Now, it is
(reasonably) well-known (Stein \cite{Stein}, Wong \cite{Wong}) that among
all $\tau$-operators only the Weyl operators enjoy a symmetry property known
as \textquotedblleft symplectic covariance\textquotedblright:

\begin{quotation}
\emph{If} $\limfunc{Op}_{\tau}(a\circ S)=\widehat{S}\limfunc{Op}_{\tau}(a)%
\widehat{S}^{-1}$ \emph{for every} $\widehat{S}\in\limfunc{Mp}(2n,\mathbb{R}%
) $ \emph{then} $\tau=\frac{1}{2}$.
\end{quotation}

\noindent Here $\limfunc{Mp}(2n,\mathbb{R})$ is the metaplectic group and $S$
the projection of $\widehat{S}\in\limfunc{Mp}(2n,\mathbb{R})$ on the
symplectic group $\limfunc{Sp}(2n,\mathbb{R})$. Symplectic covariance in the
sense above is thus a \emph{characteristic property} of Weyl
pseudo-differential calculus. In fact, one shows more generally (Stein \cite%
{Stein}, \S 12.7, Wong, Chapter 30) that:

\begin{quotation}
\emph{Let} $a\longmapsto\limfunc{Op}(a)$ \emph{be a linear mapping from} $%
\mathcal{S}^{\prime}(\mathbb{R}^{2n})$ \emph{to the space of linear
operators that is continuous in the topology of }$\mathcal{S}^{\prime}(%
\mathbb{R}^{2n})$. \emph{Assume that: (i) if }$a=a(x)$, $a\in L^{\infty}(%
\mathbb{R}^{n})$\emph{, then }$\limfunc{Op}(a)$ \emph{is multiplication by} $%
a(x)$\emph{;} \emph{(ii) if }$S\in\limfunc{Sp}(2n,\mathbb{R})$\emph{\ then} $%
\limfunc{Op}(a\circ S)=\widehat{S}\limfunc{Op}(a)\widehat{S}^{-1}$. \emph{%
Then} $a\longmapsto\limfunc{Op}(a)$ \emph{is the Weyl correspondence}
\end{quotation}

\noindent so the property of symplectic covariance really singles out Weyl
operators among all possible \textquotedblleft quantization
schemes\textquotedblright.

The principal aim of this paper is to report on the fact that there exists a
weaker form of symplectic covariance for $\tau$-operators extending which
reduces to the case above when $\tau=\frac{1}{2}$. In fact, we will show in
Proposition \ref{propro} that to each $S\in\limfunc{Sp}(2n,\mathbb{R})$ one
can attach an invertible operator $R_{\tau}(S):\mathcal{S}(\mathbb{R}%
^{n})\longrightarrow\mathcal{S}(\mathbb{R}^{n})$ such that 
\begin{equation}
R_{\tau}(S)\limfunc{Op}\nolimits_{\tau}(a)=\limfunc{Op}\nolimits_{\tau}(a%
\circ S)R_{\tau}(S).  \label{form1}
\end{equation}
These operators are in general not unitary, and do thus not generate a copy
of $\limfunc{Mp}(2n,\mathbb{R})$.

As a consequence of our constructions we will be able to establish a similar
property for Born--Jordan pseudo-differential operators. These operators
were very recently introduced in de Gosson and Luef \cite{golu1} where it
was remarked that the constructions of Boggiatto and his collaborators \cite%
{bogetal1,bogetal2,bogetal3} of a certain pseudo-differential class was
related to a quantization procedure going back to Born and Jordan \cite{bj}
and historically anterior to the work of Weyl \cite{Weyl}. Born and Jordan's
quantization is based on the prescription%
\begin{equation}
x^{m}p^{\ell}\overset{\text{\textrm{BJ}}}{\longrightarrow}\frac{1}{\ell+1}%
\sum_{k=0}^{\ell}\widehat{P}^{\ell-k}\widehat{X}^{m}\widehat{P}^{k};
\label{bj}
\end{equation}
\ \ an elementary calculation shows that this correspondence is obtained by
averaging the $\tau$-ordering (\ref{tau}) over the interval $[0,1]$. This
suggests to define more generally the Born--Jordan pseudo-differntial
operator with symbol $a$ by the formula%
\begin{equation*}
A_{\mathrm{BJ}}=\int_{0}^{1}A_{\tau}d\tau.
\end{equation*}
We will see that the symplectic covariance formula (\ref{form1}) can be used
to derive a similar formula for $A_{\mathrm{BJ}}$.

In a recent paper de Gosson and Luef \cite{golu1} have shown that this
calculus corresponds to a generalization of an early quantization scheme due
to Born and Jordan, and which has been largely superseded by the more
elegant Weyl quantization procedure. Both Weyl and Born--Jordan quantization
hark back to the early years of quantum mechanics.

\begin{notation}
The Euclidean scalar product of two vectors $u$ and $v$ on $\mathbb{R}^{m}$
is denoted indifferently $u\cdot v$ or by $uv$. When $X$ is a symmetric
matrix we will often write $Xu^{2}$ for $Xu\cdot u$. The standard symplectic
form $\sigma $ on $\mathbb{R}^{n}\times \mathbb{R}^{n}\equiv \mathbb{R}^{2n}$
is defined by $\sigma (z,z^{\prime })=px^{\prime }-p^{\prime }x$ if $z=(x,p)$%
, $z^{\prime }=(x^{\prime },p^{\prime })$ the corresponding symplectic group
is $\limfunc{Sp}(2n,\mathbb{R})$. We denote by $\mathcal{S}(\mathbb{R}^{n})$
the Schwartz space of rapidly decreasing functions on $\mathbb{R}^{n}$ and
by $\mathcal{S}^{\prime }(\mathbb{R}^{n})$ its dual (the tempered
distributions). The normalizations we use correspond to that familiar from
the theory of pseudo-differential operators; for instance the Fourier
transform of $f\in \mathcal{S}(\mathbb{R}^{n})$ is 
\begin{equation*}
Ff(x)=\int e^{-2\pi ixx^{\prime }}f(x^{\prime })dx^{\prime }
\end{equation*}%
(it corresponds to the choice $\hbar =1/2\pi $ in the quantum-mechanical
literature).
\end{notation}

\noindent\textbf{Acknowledgements.} This work has been financed by the
Austrian Research Agency FWF (Projektnummer P20442-N13).

\section{The Shubin Calculus}

\subsection{Definitions and main properties}

\subsubsection{The pseudo-differential operators $A_{\protect\tau}$}

The $\tau$-pseudo-differential operator $A_{\tau}=\limfunc{Op}_{\tau }(a)$
with symbol $a\in\mathcal{S}^{\prime}(\mathbb{R}^{n})$ is, by definition,
the operator with distributional kernel%
\begin{equation}
K_{\tau}(x,y)=F_{2}^{-1}\left[ a(\tau x+(1-\tau)y,\cdot)\right] (x-y)
\label{kertau}
\end{equation}
where $F_{2}^{-1}$ is the inverse Fourier transform in the second set of
variables. We can thus write formally (Shubin \cite{sh87}) 
\begin{equation}
A_{\tau}f(x)=\iint e^{2\pi ip(x-y)}a(\tau x+(1-\tau)y,p)f(y)dpdy.
\label{ahat1}
\end{equation}
One easily verifies using this expression that the (formal) adjoint of $%
A_{\tau}=\limfunc{Op}_{\tau}(a)$ is given by 
\begin{equation}
\limfunc{Op}\nolimits_{\tau}(a)^{\ast}=\limfunc{Op}\nolimits_{1-\tau}(%
\overline{a}).  \label{opadj}
\end{equation}

\subsubsection{The operators $\protect\widehat{T}_{\protect\tau}(z)$}

Let $\widehat{T}(z_{0})$ be the Heisenberg operator: it is defined for $f\in%
\mathcal{S}^{\prime}(\mathbb{R}^{n})$ by 
\begin{equation}
\widehat{T}(z_{0})f(x)=e^{2\pi i(p_{0}x-\frac{1}{2}p_{0}x_{0})}f(x-x_{0})
\label{hw}
\end{equation}
where $z_{0}=(x_{0},p_{0})$. Let $\tau$ be a real parameter and set, more
generally,%
\begin{equation}
\widehat{T}_{\tau}(z_{0})f(x)=e^{2\pi i(p_{0}x-(1-\tau)p_{0}x_{0})}f(x-x_{0})
\label{hopt}
\end{equation}
that is, equivalently, 
\begin{equation}
\widehat{T}_{\tau}(z_{0})=e^{i\pi(2\tau-1)p_{0}x_{0}}\widehat{T}(z_{0}).
\label{hoptbis}
\end{equation}
We have $\widehat{T}_{1/2}(z_{0})=\widehat{T}(z_{0})$, and%
\begin{equation}
\widehat{T}_{\tau}(z_{0})^{-1}=\widehat{T}_{1-\tau}(-z_{0}).  \label{tadj}
\end{equation}
It is immediate to check the following relations:%
\begin{align}
\widehat{T}_{\tau}(z_{0})\widehat{T}_{\tau}(z_{1}) & =e^{2\pi i\sigma
(z_{0},z_{1})}\widehat{T}_{\tau}(z_{1})\widehat{T}_{\tau}(z_{0})
\label{comm1} \\
\widehat{T}_{\tau}(z_{0}+z_{1}) & =e^{-i\pi\sigma(z_{0},z_{1})}\widehat{T}%
_{\tau}(z_{0})\widehat{T}_{\tau}(z_{1}).  \label{comm2}
\end{align}

For many purposes it is useful to write formula (\ref{ahat1}) in terms of
the operators $\widehat{T}_{\tau}(z)$: 
\begin{equation}
A_{\tau}f=\limfunc{Op}\nolimits_{\tau}(a)f=\int a_{\sigma}(z)\widehat{T}%
_{\tau}(z)fdz  \label{atauf}
\end{equation}
where $a_{\sigma}$ is the symplectic Fourier transform of $a$, that is 
\begin{equation*}
a_{\sigma}(z)=\int e^{-2\pi i\sigma(z,z^{\prime})}a(z^{\prime})dz^{\prime }%
\text{.}
\end{equation*}

Following the usage in the theory of Weyl operators, we will call $%
a_{\sigma} $ the \textquotedblleft twisted symbol of $A_{\tau}$%
\textquotedblright. The distributional kernel of $A_{\tau}$ can then be
written%
\begin{equation}
K_{\tau}(x,y)=F_{2}^{-1}\left[ a_{\sigma}(x-y,\cdot)\right] (\tau
x+(1-\tau)y)  \label{kertaubis}
\end{equation}
it is often more suitable for calculations than formula (\ref{kertau}).

\subsubsection{A composition formula}

The $\tau$-operators can be composed exactly in the same way as usual Weyl
operators:

\begin{proposition}
Let $A_{\tau}$ and $B_{\tau}$ be given by%
\begin{equation}
A_{\tau}=\int a_{\sigma}(z)\widehat{T}_{\tau}(z)dz\text{ \ and \ \ }B_{\tau
}=\int b_{\sigma}(z)\widehat{T}_{\tau}(z)dz.  \label{ataubetau}
\end{equation}
Then, if $A_{\tau}B_{\tau}$ is defined, we have $A_{\tau}B_{\tau}=C_{\tau}$
with%
\begin{equation}
c_{\sigma}(z)=\int e^{i\pi\sigma(z,z^{\prime})}a_{\sigma}(z-z^{\prime
})b_{\sigma}(z^{\prime})dz^{\prime}.  \label{comp1}
\end{equation}
\end{proposition}

\begin{proof}
We have%
\begin{equation*}
A_{\tau}B_{\tau}=\iint a_{\sigma}(z_{0})b_{\sigma}(z_{1})\widehat{T}_{\tau
}(z_{0})\widehat{T}_{\tau}(z_{1})dz_{0}dz_{1}
\end{equation*}
and hence, using formula (\ref{comm2})%
\begin{equation*}
A_{\tau}B_{\tau}=\iint e^{i\pi\sigma(z_{0},z_{1})}a_{\sigma}(z_{0})b_{\sigma
}(z_{1})\widehat{T}_{\tau}(z_{0}+z_{1})dz_{0}dz_{1}.
\end{equation*}
The composition formula (\ref{comp1}) follows making the change of variables 
$z=z_{0}+z_{1}$, $z^{\prime}=z$.
\end{proof}

\subsubsection{Relation with the $\protect\tau$-Wigner transform}

Boggiatto and his collaborators \cite{bogetal1,bogetal2,bogetal3} have
recently introduced a $\tau$-dependent Wigner transform $W_{\tau}(f,g)$
related with the Shubin $\tau$-pseudo-differential calculus. Averaging over $%
\tau$ in the interval $[0,1]$ leads to an element of the Cohen class, i.e.
to a transform of the type $Q(f,g)=W_{\tau}(f,g)\ast\theta$ where $\theta \in%
\mathcal{S}^{\prime}(\mathbb{R}^{2n})$.

Following result relates the operator $A_{\tau}$ to the $\tau$-Wigner
transform:

\begin{proposition}
Let $f,g\in\mathcal{S}(\mathbb{R}^{n})$. We have the formula%
\begin{equation}
(A_{\tau}f|g)_{L^{2}}=\langle a,W_{\tau}(f,g)\rangle  \label{awtau}
\end{equation}
where $W_{\tau}(f,g)$ is the $\tau$-dependent cross-Wigner transform of $%
(f,g)$ defined by%
\begin{equation}
W_{\tau}(f,g)(z)=\int e^{-2\pi iyp}f(x+\tau y)\overline{g(x-(1-\tau)y)}dy.
\label{wtau}
\end{equation}
\end{proposition}

\begin{proof}
We have%
\begin{equation*}
\langle a,W_{\tau}(\psi,\phi)\rangle=\int e^{-2\pi iyp}a(z)\psi(x+\tau y)%
\overline{\phi(x-(1-\tau)y)}dydpdx;
\end{equation*}
setting $x+\tau y=y^{\prime}$, $x-(1-\tau)y=y^{\prime}$ we get 
\begin{equation*}
\langle a,W_{\tau}(\psi,\phi)\rangle=\int e^{-2\pi
i(x^{\prime}-y^{\prime})p}a((1-\tau)x^{\prime}+\tau
y^{\prime},p)\psi(y^{\prime})\overline {\phi(x^{\prime})}dydpdx
\end{equation*}
hence the equality (\ref{awtau}) in view of (\ref{ahat1}).
\end{proof}

Formula (\ref{awtau}) yields an alternative definition of the operator $%
A_{\tau}f$ for an arbitrary symbol $a\in\mathcal{S}^{\prime}(\mathbb{R}^{n})$
and $f\in\mathcal{S}(\mathbb{R}^{n})$: choose $g\in\mathcal{S}(\mathbb{R}%
^{n})$; then$W_{\tau}(f,g)\in\mathcal{S}(\mathbb{R}^{2n})$ and the
distributional bracket $\langle a,W_{\tau}(f,g)\rangle$ is thus defined; by
definition $A_{\tau}$ is the the continuous operator $\mathcal{S}(\mathbb{R}%
^{n})\longrightarrow\mathcal{S}^{\prime}(\mathbb{R}^{n})$ defined by the
right hand-side of (\ref{awtau}).

We notice that The $\tau$-dependent Wigner transform $W_{\tau}\psi=W_{\tau
}(\phi,\psi)$ satisfies the same marginal properties as the ordinary Wigner
transform: for every $f\in L^{1}(\mathbb{R}^{n})\cap L^{2}(\mathbb{R}^{n})$
we have%
\begin{equation}
\int W_{\tau}f(x,p)dp=|f(x)|^{2}\text{ ,\ }\int W_{\tau}\psi
(x,p)dx=|Ff(p)|^{2}\text{.\ }  \label{marginal1}
\end{equation}
(see Boggiatto et al. \cite{bogetal1}).

In the case $\tau=1$ the transform $W_{\tau}$ reduces to the Rihaczek
distribution, and when $\tau=1$ we get the dual Rihaczek distribution.

\section{Symplectic Covariance in Shubin Calculus}

\subsection{A class of intertwining operators}

\subsubsection{The symplectic Cayley transform}

We will use the following notation: 
\begin{align*}
\limfunc{Sp}\nolimits_{(0)}(2n,\mathbb{R}) & =\{S\in\limfunc{Sp}(2n,\mathbb{R%
}):\det(S-I)\neq0\} \\
\limfunc{Sym}\nolimits_{(0)}(2n,\mathbb{R}) & =\{M\in \limfunc{Sym}(2n,%
\mathbb{R}):\det(M-\tfrac{1}{2}J)\neq0\}.
\end{align*}
Let $S\in\limfunc{Sp}\nolimits_{(0)}(2n,\mathbb{R})$; by definition the
symplectic Cayley transform (introduced in de Gosson \cite{Birk,JMP,RMP,JMPA}%
) of $S$ is the symmetric matrix given by 
\begin{equation}
M(S)=\tfrac{1}{2}J(S+I)(S-I)^{-1}  \label{cayley}
\end{equation}
(the symmetry of $M(S)$ is readily verified using the relation $%
S^{T}JS=SJS^{T}=J$, which is equivalent to $S\in\limfunc{Sp}(2n,\mathbb{R})$%
). The mapping $M(\cdot)$ is a bijection $\limfunc{Sp}_{(0)}(2n,\mathbb{R}%
)\longrightarrow\limfunc{Sym}\nolimits_{(0)}(2n,\mathbb{R})$ and the inverse
of that bijection is given by 
\begin{equation}
S=(M-\tfrac{1}{2}J)^{-1}(M+\tfrac{1}{2}J).  \label{invcayley}
\end{equation}
We have the properties 
\begin{equation}
M(S^{-1})=-M(S)  \label{ms}
\end{equation}
and, when in addition $S^{\prime},SS^{\prime}\in\limfunc{Sp}%
\nolimits_{(0)}(2n,\mathbb{R})$: 
\begin{equation}
M(SS^{\prime})=M(S)+(S^{T}-I)^{-1}J(M(S)+M(S^{\prime}))^{-1}J(S-I)^{-1}.
\label{mss'}
\end{equation}

\subsubsection{The intertwining operators $R_{\protect\tau}(S)$}

We will need the following well-known generalization of the Fresnel formula
(see e.g. Folland \cite{Folland}, Appendix A): let $X$ be a real invertible
matrix of dimension $m$; then: 
\begin{equation}
\int e^{-2\pi iuv}e^{i\pi Xv^{2}}dv=|\det X|^{-1/2}e^{\frac{i\pi}{4}\limfunc{%
sign}X}e^{-i\pi X^{-1}u^{2}}  \label{fresnel}
\end{equation}
where $\limfunc{sign}X$ is the difference between the number of $>0$ and $<0$
eigenvalues of $X$. Using this formula and the two lemmas above we set out
to study the operators 
\begin{equation}
R_{\tau}(S)=\sqrt{|\det(S-I)|}\int\widehat{T}_{\tau}(Sz)\widehat{T}_{\tau
}(-z)dz  \label{rs1}
\end{equation}
defined for $S\in\limfunc{Sp}\nolimits_{(0)}(2n,\mathbb{R})$.

\begin{proposition}
\label{propro}(i) Let $S\in\limfunc{Sp}\nolimits_{(0)}(2n,\mathbb{R})$. The
operator $R_{\tau}(S)$ is a continuous mapping $\mathcal{S}(\mathbb{R}%
^{n})\longrightarrow\mathcal{S}(\mathbb{R}^{n})$ satisfying 
\begin{equation}
R_{\tau}(S)\widehat{T}_{\tau}(z)=\widehat{T}_{\tau}(Sz)R_{\tau}(S)
\label{gamma2}
\end{equation}
and we have 
\begin{equation}
R_{\tau}(S)\limfunc{Op}\nolimits_{\tau}(a)=\limfunc{Op}\nolimits_{\tau}(a%
\circ S)R_{\tau}(S).  \label{interop1}
\end{equation}
(ii) Let $S$, $S^{\prime}$, $SS^{\prime}\in\limfunc{Sp}\nolimits_{(0)}(2n,%
\mathbb{R})$. We have%
\begin{equation}
R_{\tau}(SS^{\prime})=e^{i\frac{\pi}{4}\limfunc{sign}M(SS^{\prime})}R_{%
\tau}(S)R_{\tau}(S^{\prime})  \label{rs2}
\end{equation}
(iii) The operator (\ref{rs1}) satisfies 
\begin{equation}
R_{\tau}(S^{-1})=R_{\tau}(S)^{-1}=R_{1-\tau}(S)^{\ast}  \label{rs3}
\end{equation}
\end{proposition}

\begin{proof}
(i) It is equivalent to show that the operators 
\begin{equation*}
\Gamma_{\tau}(S)=\int\widehat{T}_{\tau}(Sz)\widehat{T}_{\tau}(-z)dz
\end{equation*}
are such that $\Gamma_{\tau}(S)\widehat{T}_{\tau}(z)=\widehat{T}_{\tau
}(Sz)\Gamma_{\tau}(S)$. Let $f\in\mathcal{S}(\mathbb{R}^{n})$; in view of
formula (\ref{comm2}) we have%
\begin{equation*}
\Gamma_{\tau}(S)f=\int e^{i\pi\sigma(Sz,z)}\widehat{T}_{\tau}((S-I)z)fdz;
\end{equation*}
since $S-I$ is a linear automorphism, $\widehat{T}_{\tau}((S-I)z):$ $%
\mathcal{S}(\mathbb{R}^{n})\longrightarrow\mathcal{S}(\mathbb{R}^{n})$ hence 
$\Gamma_{\tau}(S)f\in\mathcal{S}(\mathbb{R}^{n})$. The continuity of $%
\Gamma_{\tau}(S)$ is straightforward to verify. Set%
\begin{align*}
F(z,z_{0}) & =\widehat{T}_{\tau}(Sz)\widehat{T}_{\tau}(-z)\widehat{T}_{\tau
}(z_{0}) \\
G(z,z_{0}) & =\widehat{T}_{\tau}(Sz_{0})\widehat{T}_{\tau}(Sz)\widehat{T}%
_{\tau}(-z).
\end{align*}
By repeated use of formula (\ref{comm2}) one gets%
\begin{align*}
F(z,z_{0}) & =e^{-i\pi\sigma(Sz-z_{0},z-z_{0})}\widehat{T}_{\tau
}((S-I)z+z_{0}) \\
G(z,z_{0}) & =e^{-i\pi\sigma((S-I)z_{0}+Sz_{0},z)}\widehat{T}_{\tau
}((S-I)z+Sz_{0})
\end{align*}
hence $G(z-z_{0},z_{0})=F(z,z_{0})$. It follows that $\int F(z,z_{0})dz=\int
G(z,z_{0})dz$ hence the equality (\ref{gamma2}). That the operators $R_{\tau
}(S)$ satisfy the intertwining relation (\ref{interop1}) follows using
definition (\ref{atauf}) of $\limfunc{Op}\nolimits_{\tau}(a)$. (ii) (Cf. the
proof of Proposition 4.2 in de Gosson \cite{RMP}). For brevity we write $%
M=M(S)$, $M^{\prime}=M(S^{\prime})$. In view of the composition formula (\ref%
{comp1}) the twisted symbol $c_{\sigma}$ of $R_{\tau}(S)R_{\tau
}(S^{\prime}) $ is given by%
\begin{equation*}
c_{\sigma}(z)=K\int
e^{i\pi\lbrack\sigma(z,z^{\prime})+\Phi(z,z^{\prime})]}dz^{\prime}
\end{equation*}
where the constant $K$ and the phase $\Phi$ are given by 
\begin{align*}
K & =|\det(S-I)(S^{\prime}-I)|^{-1/2} \\
\Phi(z,z^{\prime}) & =Mz^{2}-2Mz\cdot z^{\prime}+(M+M^{\prime})z^{\prime2}
\end{align*}
A straightforward calculation shows that%
\begin{equation*}
\sigma(z,z^{\prime})-2Mz\cdot z^{\prime}=-2J(S-I)^{-1}z\cdot z^{\prime}
\end{equation*}
hence%
\begin{equation*}
\sigma(z,z^{\prime})+\Phi(z,z^{\prime})=-2J(S-I)^{-1}z\cdot
z^{\prime}+Mz^{2}+(M+M^{\prime})z^{\prime2}.
\end{equation*}
It follows that%
\begin{equation}
c_{\sigma}(z)=Ke^{i\pi Mz^{2}}\int e^{-2\pi iJ(S-I)^{-1}z\cdot
z^{\prime}}e^{i\pi(M+M^{\prime})z^{\prime2}}dz^{\prime}\text{.}
\label{bonin}
\end{equation}
Applying the Fresnel formula (\ref{fresnel}) with $X=M+M^{\prime}$ to the
formula above and replacing $K$ with its value we get%
\begin{equation}
c_{\sigma}(z)=|\det[(M+M^{\prime})(S-I)(S^{\prime}-I)]|^{-1/2}e^{\frac{i\pi 
}{4}\limfunc{sign}(M+M^{\prime})}e^{2\pi i\Theta(z)}  \label{cestca}
\end{equation}
where the phase $\Theta$ is given by 
\begin{align*}
\Theta(z) & =\left[ M+(S^{T}-I)^{-1}J(M+M^{\prime})^{-1}J(S-I)^{-1}\right]
z^{2} \\
& =M(SS^{\prime})z^{2}
\end{align*}
(the second equality in view of formula (\ref{mss'})). Noting that by
definition of the symplectic Cayley transform we have%
\begin{equation*}
M+M^{\prime}=J(I+(S-I)^{-1}+(S^{\prime}-I)^{-1})
\end{equation*}
it follows, using property (\ref{mss'}) of the symplectic Cayley transform,
that 
\begin{align*}
\det[(M+M^{\prime})(S-I)(S^{\prime}-I)] & =\det[(S-I)(M+M^{\prime
})(S^{\prime}-I)] \\
& =\det[(S-I)(M+M^{\prime})(S^{\prime}-I)] \\
& =\det(SS^{\prime}-I)
\end{align*}
which concludes the proof of the first part of proposition. (iii) Let us
first show that $R_{\tau}(S^{-1})=R_{\tau}(S)^{-1}$. Let $c$ be the symbol
of $C=R_{\tau}(S)R_{\tau}(S^{-1})$; we claim that $c_{\sigma}(z)=\delta(z)$,
hence $C=I$. Noting that $\det(S^{-1}-I)=\det(S-I)\neq0$, formula (\ref%
{bonin}) in the proof of part (ii) shows that 
\begin{equation*}
c_{\sigma}(z)=Le^{i\pi Mz^{2}}\int e^{-2\pi iJ(S-I)^{-1}z\cdot
z^{\prime}}e^{i\pi(M+M(S^{-1}))z^{\prime2}}dz^{\prime}
\end{equation*}
where $L=|\det(S-I)|^{-1}$. Since $M(S^{-1})=-M$ we have, setting $%
z^{\prime\prime}=(S^{T}-I)^{-1}Jz^{\prime}$,%
\begin{align*}
& =\frac{e^{i\pi Mz^{2}}}{|\det(S-I)|}\int e^{-2\pi iJ(S-I)^{-1}z\cdot
z^{\prime}}dz^{\prime} \\
& =e^{i\pi Mz^{2}}\int e^{2\pi izz^{\prime\prime}}dz^{\prime\prime}
\end{align*}
hence $c_{\sigma}(z)$ $=\delta(z)$ by the Fourier inversion formula, which
proves our claim. Let us finally show that $R_{\tau}(S^{-1})=R_{1-\tau
}(S)^{\ast}$. We have%
\begin{align*}
R_{\tau}(S^{-1}) & =\frac{1}{\sqrt{|\det(S^{-1}-I)|}}\int e^{i\pi
M(S^{-1})z^{2}}\widehat{T}_{\tau}(z)dz \\
& =\frac{1}{\sqrt{|\det(S-I)|}}\int e^{-i\pi M(S)z^{2}}\widehat{T}_{\tau
}(z)dz
\end{align*}
hence, using formula (\ref{opadj}) for the adjoint of a $\tau$-operator,%
\begin{equation*}
R_{\tau}(S^{-1})^{\ast}=\frac{1}{\sqrt{|\det(S-I)|}}\int e^{i\pi M(S)z^{2}}%
\widehat{T}_{1-\tau}(z)dz=R_{1-\tau}(S)
\end{equation*}
which is the same thing as $R_{\tau}(S^{-1})=R_{1-\tau}(S)^{\ast}$.
\end{proof}

Notice that formula (\ref{rs3}) shows that the operators $R_{\tau}(S)$ are
unitary if and only if $\tau=\frac{1}{2}$ (the Weyl case) see de Gosson \cite%
{Birk,JMP,RMP,JMPA}.

\subsubsection{Application to the $\protect\tau$-Wigner function}

The usual cross-Wigner function $W(f,g)$ has the following well-known (and
very useful) property of symplectic covariance: for all $f,g\in\mathcal{S}(%
\mathbb{R}^{n})$ and $S\in\limfunc{Sp}(2n,\mathbb{R})$ we have%
\begin{equation}
W(\widehat{S}f,\widehat{S}g)(z)=W(f,g)(S^{-1}z)  \label{wigsymp}
\end{equation}
where $\widehat{S}\in\limfunc{Mp}(2n,\mathbb{R})$ is any of the two
metaplectic operators which cover $S$. In the $\tau$-dependent case this
result must be modified as follows:

\begin{proposition}
Let $S\in\limfunc{Sp}\nolimits_{(0)}(2n,\mathbb{R})$ and $f,g\in \mathcal{S}(%
\mathbb{R}^{n})$. We have%
\begin{equation}
W_{\tau}(R_{\tau}(S)f,R_{1-\tau}(S)g)(z)=W_{\tau}(f,g)(S^{-1}z).
\label{wigsymptau}
\end{equation}
\end{proposition}

\begin{proof}
Let $A_{\tau}=\limfunc{Op}_{\tau}(a)$. Recall that $(\limfunc{Op}%
_{\tau}(a)f|g)_{L^{2}}=\langle a|W_{\tau}(f,g)\rangle$ (formula (\ref{awtau}%
). In view of the second equality (\ref{rs3}) we have, using 
\begin{align*}
(R_{\tau}(S)\limfunc{Op}\nolimits_{\tau}(a)f|R_{1-\tau}(S)g)_{L^{2}} &
=(R_{1-\tau}(S)^{\ast}R_{\tau}(S)\limfunc{Op}_{\tau}(a)f|g)_{L^{2}} \\
& =(\limfunc{Op}\nolimits_{\tau}(a)f|g)_{L^{2}} \\
& =\langle a|W_{\tau}(f,g)\rangle.
\end{align*}
On the other hand, using the intertwining property (\ref{interop1}), we have 
\begin{align*}
(R_{\tau}(S)\limfunc{Op}\nolimits_{\tau}(a)f|R_{1-\tau}(S)g)_{L^{2}} & =(%
\limfunc{Op}\nolimits_{\tau}(a\circ S)R_{\tau}(S)f|R_{1-\tau }(S)g)_{L^{2}}
\\
& =\langle a\circ S|W_{\tau}(R_{\tau}(S)f,R_{1-\tau}(S)g)\rangle \\
& =\langle a,W_{\tau}(R_{\tau}(S)f,R_{1-\tau}(S)g)\circ S^{-.1}\rangle
\end{align*}
(the last identity using the change of variables $z\longmapsto S^{-1}z$ and
the fact that $\det S=1$). Formula (\ref{wigsymptau}) follows.
\end{proof}

\subsection{The operators $R_{\protect\tau}(S)$ as pseudo-differential
operators}

Following result identifies $R_{\tau}(S)$ as a $\tau$-pseudo-differential
operator:

\begin{proposition}
\label{lemmone}Let $S\in\limfunc{Sp}\nolimits_{(0)}(2n,\mathbb{R})$; we have 
\begin{equation}
R_{\tau}(S)=\int s_{\sigma}(z)\widehat{T}_{\tau}(z)dz  \label{gamma3}
\end{equation}
where $s_{\sigma}(z)$ is given by the formula%
\begin{equation}
s_{\sigma}(z)=\frac{1}{\sqrt{|\det(S-I)|}}e^{i\pi M(S)z^{2}};  \label{gamma4}
\end{equation}
\end{proposition}

\begin{proof}
We have (see the proof of Proposition \ref{propro}) 
\begin{equation*}
R_{\tau}(S)=\int e^{i\pi\sigma(Sz,z)}\widehat{T}_{\tau}((S-I)z)dz
\end{equation*}
the change of variables $z^{\prime}=(S-I)z$ yields%
\begin{equation*}
\Gamma_{\tau}(S)=|\det(S-I)|^{-1}\int e^{i\pi\sigma(S(S-I)^{-1}z^{\prime
},z^{\prime})}\widehat{T}_{\tau}(z^{\prime})dz^{\prime}.
\end{equation*}
Since $S(S-I)^{-1}=I+(S-I)^{-1}$ we have 
\begin{align*}
\sigma(S(S-I)^{-1}z^{\prime},z^{\prime}) & =\sigma((S-I)^{-1}z^{\prime
},z^{\prime}) \\
& =J(S-I)^{-1}z^{\prime}\cdot z^{\prime} \\
& =\tfrac{1}{2}J+J(S-I)^{-1} \\
& =M(S)
\end{align*}
hence (\ref{gamma4}).
\end{proof}

For practical calculations formula (\ref{kertaubis}) is useful; it
immediately yields:

The distributional kernel of the operator $R_{\tau}(S)$ satisfies%
\begin{equation}
K_{\tau}(S)(x+y,y)=\frac{1}{\sqrt{|\det(S-I)|}}\int e^{2\pi i(\tau
x+y)p}e^{i\pi M(S)z^{2}}dp;  \label{katos}
\end{equation}
this formula can be used in principle for the calculation of explicit
expressions for the operators $R_{\tau}(S)$. Let us give an example.
Choosing $S=J$ we have $M(S)=\frac{1}{2}I$; a straightforward computation
using (\ref{katos}) yields%
\begin{equation}
K_{\tau}(J)(x,y)=e^{i\frac{n\pi}{4}}e^{\frac{i\pi}{2}(x-y)^{2}}e^{-2\pi(\tau
x+(1-\tau)y)^{2}}.  \label{ktj}
\end{equation}
Notice that when $\tau=\frac{1}{2}$ we get%
\begin{equation}
K_{1/2}(J)(x,y)=e^{i\frac{n\pi}{4}}e^{-2\pi xy}  \label{k1/2}
\end{equation}
hence $R_{1/2}(J)$ is, up to a factor, the usual Fourier transform. In fact
we have $R_{1/2}(J)\in\limfunc{Mp}(2n,\mathbb{R})$; it is the metaplectic
operator $\widehat{J}$ with projection $J$ on the symplectic group (see e.g.
de Gosson \cite{Birk}, Folland \cite{Folland}). This is not pure
coincidence, in fact:

\begin{corollary}
For $S\in\limfunc{Sp}\nolimits_{(0)}(2n,\mathbb{R})$ the operators $%
R(S)=R_{1/2}(S)$ are, up to a unimodular factor $i^{\nu(S)}$ elements of the
metaplectic group $\limfunc{Mp}(2n,\mathbb{R})$. In fact,\ when $\nu(S)$ is,
modulo $2$, the Conley--Zehnder of a any path joining the identity to $S$
then $i^{\nu(S)}R(S)\in\limfunc{Mp}(2n,\mathbb{R})$.
\end{corollary}

\begin{proof}
In \cite{Birk,JMP,RMP,JMPA} we have shown that 
\begin{equation}
R^{\nu}(S)=\frac{i^{\nu(S)}}{\sqrt{|\det(S-I)|}}\int e^{i\pi M(S)z^{2}}%
\widehat{T}(z)dz  \label{rmeta}
\end{equation}
when $\nu(S)$ is the Conley--Zehnder \cite{CZ} index (which we discuss
below). The result follows since $\widehat{T}_{1/2}(z)=\widehat{T}(z)$.
\end{proof}

\section{The Case of Born--Jordan Operators}

\subsection{Born--Jordan operators}

\subsubsection{Motivation}

Concurrently with Weyl, the physicists Born and Jordan \cite{bj} elaborated
on Heisenberg's seminal paper \cite{heisenberg} on \textquotedblleft matrix
mechanics\textquotedblright\ and proposed the quantization rule%
\begin{equation}
x_{j}^{m}p_{j}^{\ell}\overset{\text{\textrm{BJ}}}{\longrightarrow}\frac {1}{%
\ell+1}\sum_{k=0}^{\ell}\widehat{P_{j}}^{\ell-k}\widehat{X_{j}}^{m}\widehat{%
P_{j}}^{k}  \label{bj1}
\end{equation}
which coincides with (\ref{w3}) when $m+\ell\leq2$. We now make the
following fundamental remark: the Born--Jordan prescription (\ref{bj1}) is
obtained by averaging the $\tau$-ordering (\ref{tau}) on the interval $[0,1]$%
; this is immediately seen using the property 
\begin{equation*}
B(k+1,\ell-k+1)=\int_{0}^{1}(1-\tau)^{k}\tau^{\ell-k}d\tau=\frac{k!(\ell -k)!%
}{(k+\ell+1)!}
\end{equation*}
of the beta function. This suggests to study, more generally, the
pseudo-differential operators%
\begin{equation*}
A_{\mathrm{BJ}}=\int_{0}^{1}A_{\tau}d\tau.
\end{equation*}

\subsubsection{Definition of Born--Jordan operators}

In \cite{bogetal1} Boggiatto et al. define a transform $Q:\mathcal{S}(%
\mathbb{R}^{n})\times\mathcal{S}(\mathbb{R}^{n})\longrightarrow \mathcal{S}(%
\mathbb{R}^{n})$ by integrating over $[0,1]$ the $\tau$-cross Wigner
transforms (\ref{wtau}); we will use the notation $Q=W_{\text{BJ}}$; thus,
for $f,g\in\mathcal{S}(\mathbb{R}^{n})$: 
\begin{equation}
W_{\text{\textrm{BJ}}}(f,g)=\int_{0}^{1}W_{\tau}(f,g)d\tau.  \label{wbj}
\end{equation}
For $a\in\mathcal{S}^{\prime}(\mathbb{R}^{n})$ these authors define an
operator, which we denote $A_{\text{\textrm{BJ}}}$, by the formula%
\begin{equation}
(A_{\text{\textrm{BJ}}}f|g)_{L^{2}}=\langle a,W_{\text{BJ}}(f,g)\rangle
\label{abj1}
\end{equation}
$f,g\in\mathcal{S}(\mathbb{R}^{n})$ (cf. (\ref{awtau})). We will call $A_{%
\text{\textrm{BJ}}}$ the Born--Jordan operator with symbol $a$ and write $A_{%
\text{\textrm{BJ}}}=\limfunc{Op}\nolimits_{\mathrm{BJ}}(a)$.

Using the representation (\ref{atauf}) of the $\tau$-operators we have%
\begin{equation}
A_{\text{\textrm{BJ}}}=\limfunc{Op}\nolimits_{\mathrm{BJ}}(a)=\int
a_{\sigma}(z)\widehat{T}_{\mathrm{BJ}}(z)dz  \label{abj2}
\end{equation}
where $\widehat{T}_{\mathrm{BJ}}(z)$ is the unitary operator defined by 
\begin{equation}
\widehat{T}_{\mathrm{BJ}}(z)=\int_{0}^{1}\widehat{T}_{\tau}(z_{0})d\tau.
\label{tbj1}
\end{equation}
We notice that it immediately from the relations (\ref{comm1})--(\ref{comm2}%
) that:%
\begin{align}
\widehat{T}_{\mathrm{BJ}}(z_{0})\widehat{T}_{\mathrm{BJ}}(z_{1}) & =e^{2\pi
i\sigma(z_{0},z_{1})}\widehat{T}_{\mathrm{BJ}}(z_{1})\widehat{T}_{\mathrm{BJ}%
}(z_{0})  \label{comm3} \\
\widehat{T}_{\mathrm{BJ}}(z_{0}+z_{1}) & =e^{-i\pi\sigma(z_{0},z_{1})}%
\widehat{T}_{\mathrm{BJ}}(z_{0})\widehat{T}_{\mathrm{BJ}}(z_{1}).
\label{comm4}
\end{align}

\begin{proposition}
The Born--Jordan operator $A_{\text{\textrm{BJ}}}=\limfunc{Op}\nolimits_{%
\mathrm{BJ}}(a)$ is given by 
\begin{equation}
A_{\text{\textrm{BJ}}}=\int a_{\sigma}(z)\Theta(z)\widehat{T}(z)dz
\label{abj3}
\end{equation}
where $\Theta$ is the real function defined by 
\begin{equation}
\Theta(z)=\frac{\sin(2\pi px)}{2\pi px}.  \label{theta}
\end{equation}
The operator $A_{\text{\textrm{BJ}}}$ is a continuous operator $\mathcal{S}(%
\mathbb{R}^{n})\longrightarrow\mathcal{S}^{\prime}(\mathbb{R}^{n})$ for
every $a\in\mathcal{S}^{\prime}(\mathbb{R}^{2n})$. That this formula really
defines a continuous operator $\mathcal{S}(\mathbb{R}^{n})\longrightarrow 
\mathcal{S}^{\prime}(\mathbb{R}^{n})$ follows from the fact that $\Theta\in
L^{\infty}(\mathbb{R}^{2n})$.
\end{proposition}

\begin{proof}
Integrating both sides of formula (\ref{hopt}) in the interval $[0,1]$ we
have $\widehat{T}_{\mathrm{BJ}}(z)=\Theta(z)\widehat{T}(z)$ hence the
expression (\ref{abj3}).
\end{proof}

In view of the relation (\ref{opadj}) between a $\tau$-operator and its
adjoint we have 
\begin{equation}
\limfunc{Op}\nolimits_{\mathrm{BJ}}(a)^{\ast}=\limfunc{Op}\nolimits_{\mathrm{%
BJ}}(\overline{a})  \label{bjadj}
\end{equation}
hence the Born--Jordan operators share with Weyl operators the property that
they are (formally) self-adjoint if and only if their symbol is real. This
makes Born--Jordan operators good candidates for quantization.

The reader is urged to notice that while every Born--Jordan operator is a
Weyl operator, the converse property is not true because an arbitrary
distribution $b_{\sigma}\in\mathcal{S}^{\prime}(\mathbb{R}^{2n})$ cannot in
general be written in the form $a_{\sigma}\Theta$ (see de Gosson and Luef 
\cite{golu1}; also the discussion in Kauffmann \cite{kauffmann}).

\subsection{Reduced metaplectic covariance}

The intertwining properties for $\tau$ operators do not carry over to the
Born--Jordan case; it is meaningless to expect a relation like $R_{\mathrm{BJ%
}}(S)\widehat{T}_{\mathrm{BJ}}(z)=\widehat{T}_{\mathrm{BJ}}(Sz)R_{\mathrm{BJ}%
}(S)$ which would lead to a symplectic covariance property of the type (\ref%
{interop1}). The good news is, however, that Born--Jordan operators enjoy a
symplectic covariance property for operators belonging to a subgroup of the
standard metaplectic group $\limfunc{Mp}(2n,\mathbb{R})$. Recall that $%
\limfunc{Mp}(2n,\mathbb{R})$ is generated by the modified Fourier transform $%
\widehat{J}=e^{i\frac{n\pi}{4}}F$, the multiplication operators $\widehat{V}%
_{-P}f=e^{i\pi Px^{2}}f$ ($P=P^{T}$) and the unitary scaling operators $%
\widehat{M}_{L,m}f(x)=i^{m}\sqrt{|\det L|}f(Lx)$ ($\det L\neq0$, $%
m\pi=\arg\det L$). The projections of these operators on $\limfunc{Sp}(2n,%
\mathbb{R})$ are, respectively, $J$, $V_{-P}=%
\begin{pmatrix}
I & 0 \\ 
P & I%
\end{pmatrix}
$, and $M_{L}=%
\begin{pmatrix}
L^{-1} & 0 \\ 
0 & L^{2}%
\end{pmatrix}
$.

\begin{proposition}
Let $A_{\text{\textrm{BJ}}}=\limfunc{Op}\nolimits_{\mathrm{BJ}}(a)$ with $%
a\in\mathcal{S}^{\prime}(\mathbb{R}^{2n})$. We have%
\begin{equation}
\widehat{S}\limfunc{Op}\nolimits_{\mathrm{BJ}}(a)=\limfunc{Op}\nolimits_{%
\mathrm{BJ}}(a\circ S^{-1})\widehat{S}  \label{cobj1}
\end{equation}
for every $\widehat{S}\in\limfunc{Mp}(2n,\mathbb{R})$ which is a product of
a (finite number) of operators $\widehat{J}$ and $\widehat{M}_{L,m}$.
\end{proposition}

\begin{proof}
It suffices to prove formula (\ref{cobj1}) for $\widehat{S}=\widehat{J}$ and 
$\widehat{S}=\widehat{M}_{L,m}$. Let $\widehat{S}$ be anyone of these
operators; we have 
\begin{align*}
\widehat{S}\limfunc{Op}\nolimits_{\mathrm{BJ}}(a) & =\int a_{\sigma
}(z)\Theta(z)\widehat{S}\widehat{T}(z)dz \\
& =\left( \int a_{\sigma}(z)\Theta(z)\widehat{T}(Sz)dz\right) \widehat{S}
\end{align*}
where the second equality follows from the usual symplectic covariance
property $\widehat{S}\widehat{T}(z)=\widehat{T}(Sz)\widehat{S}$ of the
Heisenberg operators. Making the change of variables $z^{\prime}=Sz$ in the
integral we get, since $\det S=1$,%
\begin{equation*}
\int a_{\sigma}(z)\Theta(z)\widehat{T}(Sz)dz=\int
a_{\sigma}(S^{-1}z)\Theta(S^{-1}z)\widehat{T}(z)dz.
\end{equation*}
Now, by definition of the symplectic Fourier transform we have%
\begin{equation*}
a_{\sigma}(S^{-1}z)=\int e^{-2\pi i\sigma(S^{-1}z,z^{\prime})}a(z^{\prime
})dz^{\prime}=(a\circ S^{-1})_{\sigma}(z).
\end{equation*}
On the other hand 
\begin{equation*}
\Theta(M_{L}^{-1}z)=\frac{\sin(2\pi Lp\cdot(L^{T})^{-1}x)}{2\pi Lp\cdot
(L^{T})^{-1}x}=\Theta(z)
\end{equation*}
and, similarly, $\Theta(J^{-1}z)=\Theta(z)$ so we have%
\begin{align*}
\widehat{S}\limfunc{Op}\nolimits_{\mathrm{BJ}}(a) & =\left( \int(a\circ
S^{-1})_{\sigma}\Theta(z)\widehat{T}(z)dz\right) \widehat{S} \\
& =\limfunc{Op}\nolimits_{\mathrm{BJ}}(a\circ S^{-1})\widehat{S}
\end{align*}
whence formula (\ref{cobj1}).
\end{proof}

The proof above shows that the essential step consists in noting that $%
\Theta(S^{-1}z)=\Theta(z)$ when $S=J$ or $S=M_{L}$. It is clear that this
property fails if one takes $S=V_{P}$ with $P\neq0$, so we cannot expect to
have full symplectic covariance for Born--Jordan operators. Such a property
is anyway excluded in view of our discussion in the Introduction to this
paper. symplectic covariance is characteristic of Weyl calculus.


\begin{thebibliography}{99}
\bibitem{bogetal1} Boggiatto P., De Donno G., Oliaro A., Time-Frequency
Representations of Wigner Type and Pseudo-Differential Operators,
Transactions of the Amer. Math. Soc., \textbf{362}(9) 4955--4981 (2010)

\bibitem{bogetal2} Boggiatto P., Bui Kien Cuong, De Donno, G., Oliaro A.,
Weighted integrals of Wigner representations, Journal of Pseudo-Differential
Operators and Applications, 2010

\bibitem{bogetal3} Boggiatto P., De Donno G., Oliaro A., Hudson Theorem for $%
\tau$-Wigner Transforms, Quaderni scientifici del Dipartimento di
Matematica, Univerist\`{a} di Torino, Quaderno N. 1, 2010

\bibitem{bj} Born, M., Jordan, P., Zur Quantenmechanik, Z. Physik 34,
858--888 (1925)

\bibitem{CZ} Conley C., Zehnder, E., Morse-type index theory for flows and
periodic solutions of Hamiltonian equations. \textit{Comm. Pure and Appl.
Math. }\textbf{37} (1984) 207--253

\bibitem{Folland} Folland, G.B., Harmonic Analysis in Phase space. Annals of
Mathematics studies, Princeton University Press, Princeton, N.J. (1989)

\bibitem{lett} de Gosson M., The Weyl Representation of Metaplectic
operators. \textit{Letters in Mathematical Physics} \textbf{72 }(2005)
129--142

\bibitem{Birk} Symplectic Geometry and Quantum Mechanics, Birkh\"{a}user,
Basel, series \textquotedblleft Operator Theory: Advances and
Applications\textquotedblright\ (subseries: \textquotedblleft Advances in
Partial Differential Equations\textquotedblright), Vol. 166 (2006)

\bibitem{JMP} de Gosson M., de Gosson, S., An extension of the
Conley--Zehnder Index, a product formula and an application to the Weyl
representation of metaplectic operators . \textit{J. Math. Phys., }\textbf{47%
}(12), 123506, 15 pages, (2006)

\bibitem{RMP} de Gosson M., Metaplectic Representation, Conley--Zehnder
Index, and Weyl Calculus on Phase Space. \textit{Rev. Math. Physics}, 
\textbf{19}(8) (2007) 1149--1188

\bibitem{JMPA} de Gosson M., On the usefulness of an index due to Leray for
studying the intersections of Lagrangian and symplectic paths. J. Math.
Pures Appl. \textbf{91} (2009) 598--613

\bibitem{golu1} de Gosson, M., Luef, F., Preferred Quantization Rules:
Born--Jordan vs. Weyl; Applications to Phase Space Quantization. \textit{J.
Pseudo-Differ. Oper. Appl.}, 2(1) (2011) 115--139

\bibitem{heisenberg} Heisenberg, W., \"{U}ber quantentheoretische Umdeutung
kinematischer und mechanischer Beziehungen, Zeitschrift f\"{u}r Physik, 33,
879--893 (1925) [English translation in: B.L. van der Waerden, editor,
Sources of Quantum Mechanics, Dover Publications (1968)

\bibitem{kauffmann} Kauffmann, S.K. , Unambiguous quantization from the
maximum classical correspondence that is self-consistent: the slightly
stronger canonical commutation rule Dirac missed, Found. Phys. 41, 805--918
(2011)

\bibitem{Leray} Leray, J., \textit{Lagrangian Analysis and Quantum Mechanics}%
,\textit{\ a mathematical structure related to asymptotic expansions and the
Maslov index} (the MIT Press, Cambridge, Mass., 1981); translated from 
\textit{Analyse Lagrangienne} RCP 25, Strasbourg Coll\`{e}ge de France,
1976--1977

\bibitem{reiter} Reiter, H., Metaplectic Groups and Segal Algebras, Springer
(1989)

\bibitem{sh87} Shubin, M.A., Pseudodifferential Operators and Spectral
Theory, Springer--Verlag, (1987) [original Russian edition in Nauka, Moskva
(1978)]

\bibitem{Stein} Stein, E.M., \textit{Harmonic Analysis: Real Variable
Methods, Orthogonality, and Oscillatory Integrals}. Princeton University
Press (1993)

\bibitem{Weil} Weil, A., Sur certains groupes d'op\'{e}rateurs unitaires, 
\textit{Acta Math}. \textbf{111}, 143--211 (1964); also in \textit{Collected
Papers}, Vol. III:1--69, Springer-Verlag, Heidelberg, 1980.

\bibitem{Weyl} Weyl, H., Quantenmechanik und Gruppentheorie, Zeitschrift f%
\"{u}r Physik, 46 (1927)

\bibitem{Wong} Wong, M.W., Weyl Transforms, Springer (1998)
\end{thebibliography}
\end{document}